\documentclass{article}

\usepackage{booktabs} 
\usepackage{amsfonts,amssymb,amsmath,latexsym,ae,aecompl,bbm,algorithm,amsthm}
\usepackage[noend]{algpseudocode}
\usepackage{fullpage}
\usepackage{color}

\newtheorem{theorem}{Theorem}

\newtheorem{lemma}{Lemma}

\newcommand{\magenta}[1]{\textcolor{magenta}{#1}}
\def\ShowComment{True}
\ifdefined\ShowComment
\def\john#1{\marginpar{$\leftarrow$\fbox{J}}\footnote{$\Rightarrow$~{\sf #1 \magenta{--John}}}}
\else
\def\john#1{}
\fi

\begin{document}
\title{Dispersion of Mobile Robots: A Study of Memory-Time Trade-offs 
\thanks{A preliminary version of this paper was accepted to the International Conference on Distributed Computing and Networking 2018~\cite{AM18}.}}

\author{John Augustine\thanks{Department of Computer Science \& Engineering, Indian Institute of Technology Madras, Chennai, India. augustine@iitm.ac.in. Research supported in part by an Extra-Mural Research Grant (file number EMR/2016/003016) funded by the Science and Engineering Research Board, Department of Science and Technology, Government of India.}
\and William K. Moses Jr.\thanks{Department of Computer Science and Engineering, Indian Institute of Technology Madras, Chennai, India. wkmjr3@gmail.com. Research supported in part by an Extra-Mural Research Grant (file number EMR/2016/003016) funded by the Science and Engineering Research Board, Department of Science and Technology, Government of India.}
}

\date{}

\maketitle

\begin{abstract}

We introduce a new problem in the domain of mobile robots, which we term dispersion. In this problem, $n$ robots are placed in an $n$ node graph arbitrarily and must coordinate with each other to reach a final configuration such that exactly one robot is at each node. We study this problem through the lenses of minimizing the memory required by each robot and of minimizing the number of rounds required to achieve dispersion. 

Dispersion is of interest due to its relationship to the problems of scattering on a graph, exploration using mobile robots, and load balancing on a graph. Additionally, dispersion has an immediate real world application due to its relationship to the problem of recharging electric cars, as each car can be considered a robot and recharging stations and the roads connecting them nodes and edges of a graph respectively. Since recharging is a costly affair relative to traveling, we want to distribute these cars amongst the various available recharge points where communication should be limited to car-to-car interactions.

We provide lower bounds on both the memory required for robots to achieve dispersion and the minimum running time to achieve dispersion on any type of graph. We then analyze the trade-offs between time and memory for various types of graphs. We provide time optimal and memory optimal algorithms for several types of graphs and show the power of a little memory in terms of running time.
\end{abstract}



\section{Introduction}
\label{sec:intro}
\subsection{Background \& Motivation}
The use of mobile robots to solve global problems in a distributed manner is a new and interesting paradigm in problem solving. In it, each robot acts individually, but collectively the robots accomplish some goal that would be infeasible to solve using a global centralized approach. Many important real world problems such as toxic hazard clean-up, large maze exploration, and gathering at one place can be modeled in this paradigm. We introduce a new problem in the domain of mobile robots, which we term dispersion. In this problem, $n$ robots are placed in an $n$ node graph arbitrarily and must coordinate with each other to reach a final configuration such that exactly one robot is at each node. 

A practical application of this idea comes from the area of self-driving electric cars and recharging stations. Typically recharging an electric car is a time-intensive task (as public recharging centers take hours) and when multiple recharge stations are located nearby, it is better in terms of time to simply find the nearest free station instead of waiting. With self-driving cars becoming a reality, it is advantageous to figure out strategies to automate this process and furthermore leverage knowledge, if available, about the larger spread of available stations.

We now look at the relation of dispersion to other well-known problems in literature, namely scattering, exploration with mobile robots, and load balancing.

This problem is very similar to the robot scattering or uniform-deployment problem on graphs \cite{BFMS11,EB11,SMOKM16}, where robots must spread out evenly in the graph. When the number of robots is equal to the number of nodes in the graph, the problems are exactly the same assuming the model is the same.

Dispersion is also very similar to the problem of $n$ robot collaborative graph exploration using mobile robots \cite{DLS07,BCGX11,BVX14,DDKPU15}, where $n$ robots start at a given node and must explore the graph in as few rounds as possible. Since any solution to a problem of dispersion acts as a solution to $n$ robot exploration under the same conditions, any possibility results achieved for dispersion apply to $n$ robot exploration. Furthermore, results in exploration are usually of two types, either trying to achieve exploration in a ring or tree in as few rounds as possible, or else trying to show the possibility of exploration with as few bits as possible. We feel that by showing the interplay of memory and running time for not just rings and trees, but also arbitrary graphs, we add meaningful contributions to that area of work.

This problem may also be considered a variant of load balancing on graphs \cite{BV86,Cybenko89,PV89,SS94,MGS98,BFH09,SS12,BCFFS15}, where typically the nodes start with some arbitrary amount of load and must transfer load using edges until each node has more or less the same amount. We can think of dispersion in this context if we assume that the loads have memory, computational power, and are labeled while the nodes do not have memory or computational power and are unlabeled. This is of interest as studying load balancing in this way acts as a sort of bridge between the vastly different fields of mobile robots and load balancing and may pave the way to an exchange of ideas and techniques between the two areas.


\subsection{Related Work}
Dispersion is similar to the problem of scattering or uniform deployment of $k$ robots on an $n$ node graph when $k=n$. Scattering requires $k$ robots to uniformly deploy themselves in a given network with $n$ nodes. \cite{BFMS11} looked into this problem on grids while \cite{EB11,SMOKM16} looked into this problem on rings under different settings than the current paper.

The problem of dispersion is very close to that of $k$ robot collective exploration of an $n$ node graph when $k=n$. The problem states that, given $k$ robots initially all located on a given node, we want an algorithm run by all robots such that all nodes are visited in the least amount of time. Typically the nodes are labeled. Much work has been done on this problem \cite{DLS07,BCGX11,BVX14,DDKPU15}, especially for rooted trees \cite{DKHS06,FGKP06,BCGX11,HKLT14,OS14}. Of particular interest to us are the results of \cite{OS14} and \cite{DKHS06}, which both use $k$ robots to explore a rooted tree in $O(D^2)$ rounds and $O(D^{2 - 1/p})$ rounds respectively, where $p$ is a property of the graph called its ``density". \cite{DMNSS17} gives a lower bound for $n$ robot collaborative tree exploration as $\Omega(D^2)$ rounds.  \cite{DKHS06} achieves a bound better than $O(D^2)$ by requiring robots to have knowledge of the value of $p$ and leveraging that knowledge. Furthermore, under their setting each node is labeled and robots have unlimited memory. 

The notion of limiting the robots' memory was raised in \cite{DFKP02} where they proved that exploration with stopping on a tree with unlabeled nodes by a single robot was possible with $O(\log^2 n)$ bits of memory. \cite{AGPRZ11} improved this by showing that it was possible to explore with stopping a tree with unlabeled nodes with one robot having $O(\log n)$ bits of memory. Other attempts at limiting robots memory were also explored in \cite{CFIKP08,MPU15,SBNOKM15} where nodes were allowed to have some memory as well. \cite{CFIKP08} allowed nodes to have $1-3$ bits and showed that it was possible for a single robot to explore such graphs subject to certain constraints. \cite{MPU15} studied exploration of a graph under a slightly different model by a single robot when both robots and nodes were allowed some bits of memory. Their main result was an exploration algorithm that required each robot to have $1$ bit of memory and each node to have $O(\log \Delta)$ bits of memory in order to achieve exploration in $O(m)$ time steps. \cite{SBNOKM15} present various algorithms which trade-off memory of agents and memory of nodes.

Load balancing requires a total amount of load to be distributed among several processors. If we consider the robots to be the load, then dispersion is equivalent to load balancing on a graph when ``smart" loads are in play, i.e. the loads make the decisions on where to move and not the nodes. Loads are either discrete \cite{BV86,PV89} or continuous \cite{MGS98}. Dispersion is closer to load balancing when discrete loads are involved. Load balancing in graphs has been studied for quite a while now \cite{BV86,Cybenko89,PV89,SS94,MGS98,BFH09,SS12,BCFFS15}. Work in the area is usually divided into either work dealing with diffusion \cite{Cybenko89,SS94,MGS98} or dimension exchange \cite{XL92}, which refers to whether a node can balance load with all its neighbors concurrently or just one neighbor at a time in a given round respectively. Our model is closer to the work done on diffusion.

The problem of dispersion, in the setting of time-varying rings, is looked into in the work of \cite{AAMSS18}.

\paragraph{Note:} The current full version of this paper is slightly different from the conference version of the paper~\cite{AM18}. We improve the algorithm to achieve dispersion on a rooted tree from requiring $O(\log n + \Delta)$ bits of memory to just requiring $O(\log n)$ bits of memory. Moreover, in this paper we correct some  errors in the dispersion algorithms for general graphs that appeared in the conference version. Interestingly these corrections have only improved upon the bounds achieved in the conference version.
We note that Kshemkalyani and Ali~\cite{KA18} mentioned these errors and provided their independent corrections. 

\subsection{Our Results}
We provide both lower bounds and upper bounds for the problem of deterministic dispersion on various types of graphs for different memory constraints on the individual robots. A list of our upper bound results may be found in Table~\ref{table:results}. We assume that each robot has no visibility of the graph and can only communicate with other robots present on their own node. Furthermore, we assume that robots only know the number of nodes $n$ and number of edges $m$ of the graph, but do not know the maximum degree $\Delta$ or the diameter $D$. As any solution to dispersion also solves collaborative graph exploration and also scattering under the same conditions, our algorithms apply to those problems as well. This is especially interesting, as to the best of our knowledge, ours is the first work analyzing time and memory trade-offs across various types of graphs for exploration as opposed to merely discussing the possibility of achieving exploration with limited memory.

\begin{table*}[ht]
	\caption{Upper bound results for different types of graphs along with the memory requirement of each robot. For a given graph, $n$ is number of nodes, $m$ is number of edges, and $D$ is diameter. }
	\centering \vspace{1em}
	\begin{tabular}{|c|c|c|c|c|}
		\hline
		Serial No. & Type of Graph & Memory Requirement  & Algorithm Name & Time Until\\
		& & of Each Robot & &  Dispersion Achieved  \\
		\hline
		\hline
		1. & Path & $O(\log n)$ bits & Path-Ring-Tree-LogN & $O(n)$ rounds \\
		\hline
		2. & Ring & $O(\log n)$ bits & Path-Ring-Tree-LogN & $O(n)$ rounds \\
		\hline
		3. & Tree & $O(\log n)$ bits & Path-Ring-Tree-LogN & $O(n)$ rounds \\
		\hline
		4. & Rooted Tree & $O(\log n)$ bits & Rooted-Tree-LogN & $O(D^2)$ rounds \\
		\hline
		5. & Arbitrary Graph & $O(\log n)$ bits & Graph-LogN & $O(mn)$ rounds \\
		\hline
		6. & Rooted Graph & $O(\log n)$ bits & Rooted-Graph-LogN & $O(m)$ rounds \\
		\hline
		7. & Arbitrary Graph & $O(n \log n)$ bits & Graph-N-LogN & $O(m)$ rounds \\
		\hline
	\end{tabular}
	\label{table:results}
\end{table*}

Given any graph, we show that each robot requires at least $\Omega(\log n)$ bits of memory in order to achieve deterministic dispersion when all robots have the same amount of memory. It is easy to see that any algorithm will take at least $\Omega(D)$ rounds to achieve dispersion on a graph. We are able to show a stronger bound of $\Omega(n)$ rounds for arbitrary graphs.

Initially, we restrict robots to have $O(\log n)$ bits of memory. We develop the algorithm \emph{Path-Ring-Tree-LogN} for paths, trees, and rings. It takes $O(n)$ rounds to achieve deterministic dispersion. It is asymptotically optimal in both memory and time for paths and rings but is only asymptotically optimal in memory for trees. We develop the algorithm \emph{Rooted-Graph-LogN} for rooted graphs, i.e. graphs where all robots start at one node. It is asymptotically optimal in terms of memory and requires $O(m)$ rounds to achieve dispersion. We then construct \emph{Graph-LogN} which achieves dispersion on arbitrary graphs in $O(mn)$ rounds is asymptotically optimal in terms of memory. We then develop \emph{Rooted-Tree-LogN} for rooted trees, i.e. trees where all robots start at one node. It takes $O(D^2)$ rounds to achieve dispersion. For the given constraint of lack of knowledge of the exact topology of the tree, the algorithm appears to be asymptotically optimal in time.\footnote{\cite{DMNSS17} claims (Theorem 2.5 in the arXiv version of that paper) that the lower bound for $n$ robot collaborative tree exploration is $\Omega(D^2)$ rounds, rendering our algorithm asymptotically time optimal.}

Finally, we present the algorithm \emph{Graph-N-LogN}, which requires $O(n \log n)$ bits of memory and achieves dispersion on any graph in $O(m)$ rounds. This algorithm illustrates the power of trading off memory to achieve faster and more powerful algorithms.


\subsection{Organization of Paper}
The technical preliminaries are presented in Section~\ref{sec:prelims}. Lower bounds for the problem are given in Section~\ref{sec:lower-bounds}. Our results for paths, rings, trees, rooted graphs, graphs, and rooted graphs when robots have $O(\log n)$ bits of memory are presented in Section~\ref{sec:logn-memory}. Our result for an arbitrary graph when robots have $O(n \log n)$ bits of memory is presented in Section~\ref{sec:n-logn-memory}. Conclusions and some open problems are presented in Section~\ref{sec:conc}.


\section{Technical Preliminaries}
\label{sec:prelims}

We now define the model and subsequently formally state the problem description.

\subsection{Model}
\textbf{Parameters of graph:}\\ 
We consider a graph with $n$ nodes, $m$ edges, maximum degree (of any node) $\Delta$, and diameter $D$. The edges are unweighted and undirected. Nodes are anonymous, i.e. they don't have unique ID's. For each node, a unique integer in $[0, \delta-1]$, where $\delta$ is the degree of the node, called port number is assigned to each edge connected to that node. Edges can be thought of as bridges between two nodes, where each node has its own port to denote the bridge. We do not assume any relation between the two port numbers for a given edge. Any number of robots are allowed to move along an edge in a given round.\\
\\
\textbf{Type of communication:}\\
We consider a synchronous system where rounds are counted according to a global clock.
Every round consists of two steps. In the first step, any robot can communicate with other robots co-located with it and perform local computations\footnote{Note that for a given robot with $O(x)$ bits of memory, the local computation will not exceed an $O(x)$ space computation. For example, if all robots only have $O(\log n)$ bits of memory, then any robot can only perform local computations in $\log n$ space.}. In the second step, if the robot has decided to move along a given edge, it will then move along it. We do not restrict the time for local computations of the robots and communication that can take place amongst themselves.\\
\\
\textbf{Powers of robots:}\\
Each robot has a unique label taken from the range $[1, c \log n]$, where $c \geq 1$ is a constant. Each robot has no visibility of the graph and can only communicate with other robots present on the same node using messages. When a robot moves from node $u$ to node $v$, it is aware of the port it used to leave $u$ and the port it used to enter $v$. Robots only know the values of $n$ and $m$ and don't know the values of $D$ and $\Delta$.\\
\\
\textbf{Complexity measures:}\\
We measure the efficiency of our algorithms against the number of rounds taken for dispersion to occur on the underlying graph.

Another important parameter we would like to understand is the memory complexity required by each robot. This is the number of bits that each robot is allowed to use to store information and perform local computation.\\
\\\\
\textbf{Definitions:}\\ We use the notion of \textit{assigning a robot to a node}. When we assign a given robot to a given node, we imply that in the final dispersed configuration that robot will be present on that node. We also call a node an \textit{assigned node} or \textit{unassigned node} when a robot is assigned to it or not respectively.

\subsection{Problem Description}
The problem of dispersion is that given an initial configuration of $n$ robots arbitrarily assigned to an $n$ node graph, we want robots to move around such that we arrive at a configuration where exactly one robot is present on each node. We attempt to come up with algorithms to perform dispersion on various types of graphs while minimizing the number of rounds taken to achieve dispersion and minimizing the memory requirement for each robot.

\section{Lower Bounds on Time \& Memory}
\label{sec:lower-bounds}
\subsection{Lower Bound on Running Time}
It is clear to see that the lower bound for any algorithm to perform dispersion on any graph is $\Omega(D)$. This is because of the initial configuration where all robots are located on the same node and there exists a node at a distance of $D$ away from that node. It takes at least $D-1$ rounds to reach that node. Note that in the case of paths and rings, $D = \Theta(n)$, so we have a lower bound of $\Omega(n)$ on running time. 

From \cite{DMNSS17}, we have a lower bound for trees and arbitrary graphs as $\Omega(D^2)$ rounds. This is because they prove their lower bound for $n$ robot collaborative tree exploration, and we can use any algorithm to solve dispersion as a sub-algorithm to solve $n$ robot collaborative tree exploration, ensuring that their bound applies to dispersion for trees.

We present below a lower bound of $\Omega(n)$ rounds for running time in the case of arbitrary graphs. 

\begin{theorem}\label{the:lower-bound-arb-graphs}
Any  deterministic dispersion algorithm to solve dispersion on arbitrary graphs requires $\Omega(n)$ rounds in the worst case.
\end{theorem}

\begin{proof}
We construct a graph that we term a \textit{dumbbell graph} and show how no matter what the robots do, they can never achieve dispersion in $o(n)$ rounds. The dumbbell graph is similar to an $n$-barbell graph but requires two bridges from a clique and ensures that every node within a clique has the same degree. We construct the graph as follows. If $n$ is even, divide the nodes into two cliques $a_1, \ldots, a_{n/2}$ and $b_1, \ldots, b_{n/2}$. Remove the edges $(a_1,a_{n/2})$ and $(b_1, b_{n/2})$ and add edges $(a_1,b_1)$ and $(a_{n/2}, b_{n/2})$. Call the added edges \textit{bridges} as they act as bridges between the two cliques. If $n$ is odd, divide the nodes into two cliques $(a_1, \ldots, a_{\lfloor n/2 \rfloor})$ and $(b_1, \ldots, b_{\lfloor n/2 \rfloor})$ and a single node $c$. Remove the edges $(a_1, a_{\lfloor n/2 \rfloor})$ and $(b_1, b_{\lfloor n/2 \rfloor})$ and add the edges $(a_1, c)$, $(a_{\lfloor n/2 \rfloor}, c)$, $(b_1, c)$, and $(b_{\lfloor n/2 \rfloor}, c)$. We again call the added edges bridges.

Let us slightly alter the model to benefit the algorithm designer thus making lower bounds harder to achieve. We assume that nodes have unique identities. We assume that a central authority can coordinate all the robots and tell them where to move each round. The robots follow this set of movements and in turn update the central authority with the names of the nodes seen. Furthermore, we assume that the central controller knows which nodes are in which cliques (and in the case when $n$ is odd, which node node is a stand-alone).

Now, we continue with the proof. Let us assume that all robots start at one node in one of the cliques. We show that for every deterministic algorithm, there exists a dumbbell graph such that it takes $\Omega(n)$ rounds before any robot will cross one of the two bridges from the clique. Since dispersion requires all nodes to be settled, this shows that it takes $\Omega(n)$ rounds to achieve dispersion in the graph. We constructed these dumbbell graphs such that all nodes within a given clique have the same degree. Thus, since the nodes are labelled and a central controller controls how the robots move, the only piece of information not known to the controller is what ports on the nodes lead to what edges. Thus the controller will try to explore the maximum number of edges in each round. Since there are $n$ robots, the controller can explore at most $n$ new edges each round. There are a total of at least $\lfloor n/2 \rfloor * (\lfloor n/2 \rfloor - 1)/2$ edges to explore, two of which act as bridges. Thus it takes $\Omega(\lfloor n/2 \rfloor * (\lfloor n/2 \rfloor - 1)/2) / n) = \Omega(n)$ rounds to explore all edges. There exists a dumbbell graph for every deterministic algorithm that guarantees that the first bridge is explored in either the last round or the second to last round.
\end{proof}

\subsection{Lower Bound on Memory of Robots}
\begin{theorem}
Assuming all robots are given the same amount of memory, robots require $\Omega(\log n)$ bits of memory each for any deterministic algorithm to achieve dispersion on a graph. 
\end{theorem}

\begin{proof}
We prove this theorem by showing that if robots have $o(\log n)$ bits of memory each, then dispersion is impossible.

Suppose all robots have $o(\log n)$ bits of memory. Each robot's state space is then $2^{o(\log n)} = n^{o(1)}$. Since there are $n$ robots, by pigeonhole principle, there exist two robots $u$ and $v$ with the same state space.

Let us suppose that that all robots are initially on the same node. Since all robots run the same deterministic algorithm and $u$ and $v$ are co-located initially, they will perform the same moves. In essence, they can be considered ``sticky" robots, in that they will always mirror each other's move and will never do anything different. Since dispersion requires a configuration where there is exactly one robot per node, we will never achieve dispersion since there is no way for $u$ and $v$ to settle down on two different nodes.

Thus, assuming all robots have the same amount of memory, robots require $\Omega(\log n)$ bits of memory each in order to achieve dispersion on a graph.
\end{proof}

\section{Dispersion with $O(\log n)$ Bits of Memory}
\label{sec:logn-memory}

When each robot has $O(\log n)$ bits of memory, we look at dispersion on rings, trees, rooted graphs, and rooted trees. Three algorithms presented in this section are variants of basic depth-first search while the algorithm for rooted tree is a bit different.

\subsection{Dispersion on a Path, Ring, or Tree}\label{subsec:path-ring-tree-logn}
We present an $O(n)$ round dispersion algorithm for paths, rings, or trees. Note that since the lower bound on running time of dispersion for rings and paths is $\Omega(n)$, this algorithm is asymptotically optimal for those two types of graphs.

The algorithm \emph{Path-Ring-Tree-LogN} works by having each robot see if it can get assigned to the node it is on at the start of every round. If not and it entered the node through port $i$, it leaves through port $(i+1) \mod \text{ degree of node}$. If a robot is at a leaf node or one end of a path, then it exits the node through the same port it entered. Note that this algorithm closely mirrors the algorithm proposed by Amb{\"u}hl et al. \cite{AGPRZ11} to solve single robot exploration in a tree with $O(\log n)$ bits of memory. 

The algorithm uses three variables, $port\_entered$, $settled$, and $rnd\_cntr$. $port\_entered$ denotes the port through which a robot entered the current node. $settled$ is a boolean variable which indicates if the robot has assigned itself to a node or not. $rnd\_cntr$ is used to keep track of the rounds.

Note that though we will prove that the algorithm achieves dispersion in $\leq 2n$ rounds, we still require each robot to maintain a round counter and count up to $2n$ rounds before explicitly terminating. This is required because every time an exploring robot moves to a new node, it must confirm that that node doesn't already have another robot settled on it. This is done by having all co-located robots communicate in a round and establish if one robot has already settled there or not. If a robot terminates early, it may happen that another robot may also choose to settle at that node later on and then dispersion can never be achieved. This logic applies to many of the algorithms in the subsequent sections with the exception of  algorithm \emph{Rooted-Tree-LogN} in Section~\ref{subsec:rooted-tree-logn}. In that algorithm, robots choose when to terminate, not based on a round counter.

\alglanguage{pseudocode}
\begin{algorithm}
	\caption{Path-Ring-Tree-LogN, run by each robot $u$}
	\label{prot:Path-Ring-Tree-LogN}
	\begin{algorithmic}[1]
		\State Initialize the following values: $port\_entered \leftarrow 0$, $settled \leftarrow FALSE$.
		\For{$rnd\_cntr \gets 1, 2n$}
			\State Set $port\_entered$ to the port entered through.
			\If {Node doesn't have robot assigned to it and $u$ is robot with lowest label presently on node}
				\State $u$ assigns itself to node and sets $settled \leftarrow TRUE$.
			\EndIf
			\If {$settled = FALSE$}
				\State $port\_entered \leftarrow (port\_entered + 1) \mod$ degree of node.
				\State Move through $port\_entered$.
			\EndIf
		\EndFor
		\Statex
	\end{algorithmic}
\end{algorithm}

\begin{theorem} \label{the:path-ring-tree-logn}
	Algorithm \emph{Path-Ring-Tree-LogN} can be run by robots with $O(\log n)$ bits of memory to ensure dispersion occurs in $O(n)$ rounds on paths, rings and trees.
\end{theorem}

\begin{proof}
	Each robot performs the algorithm for $2n$ rounds, so it easy to see that the time complexity is $O(n)$. The memory limit on robots is easy to see as we only need robots to remember the port through which they entered ($O(\log \Delta)$ bits), the current round ($O(\log n)$ bits) and if they are settled ($1$ bit). Additionally, all robots require $O(\log n)$ bits to compare their labels with other robots in order to decide which robots are assigned to which nodes for a total of $O(\log n)$ bits of memory.
	
	We now need to show that dispersion is achieved in $2n$ rounds. It is clear to see that no more than one robot will be assigned to any given node because robot labels are unique and once one robot has been assigned to a given node, no other robot may then be assigned to it. All that remains is to argue that in $2n$ rounds each robot is assigned to a node. This is achieved if we can upper bound the time taken by any robot to visit all nodes as $2n$. In a path, ring or tree, the number of edges is at most $n$. It takes at most $2n$ rounds to perform depth-first search on these structures. Since \emph{Path-Ring-Tree-LogN} is essentially performing depth-first search, we have it that all nodes will be visited by a given robot if it hasn't settled in at most $2n$ rounds.
\end{proof}

\subsection{Dispersion on a Rooted Graph}\label{subsec:logn-disp-rooted-graph}

A rooted graph is a graph where all robots are initially located at the same node. We are able to achieve dispersion in $O(m)$ rounds using ideas we used in the prior section and a technique to prevent robots from getting caught in cycles. A robot will get caught in a cycle if it is exploring the graph and cannot figure out if it has already visited a node or not. Note that this problem arises only due to our constraint on every robot's memory. 

Our technique to solve this problem is as follows. Since all robots follow the same path, it is easy to tell if we have previously visited a given node by checking if that node has a robot already assigned to it. If so, the exploring robot can backtrack and visit some other node. However, a robot must be able to discern whether it is exploring or backtracking. In order to know which one it is doing, we require every settled robot to maintain a pointer to the port it used to enter its assigned node. This acts as a parent pointer. If some other exploring robot leaves the node through this same port, then it is backtracking. This can easily be found out by having every settled node transmit its parent pointer to other robots co-located on the same node as it during the communication part of a round. The reason we require all robots to start at the same node initially is because this method to keep track of backtracking fails if we have robots starting at multiple places.

The algorithm uses four variables, $port\_entered$, $parent\_ptr$, $state$, and $rnd\_cntr$. $port\_entered$ denotes the port through which a robot entered the current node. $parent\_ptr$ is used by a settled robot to indicate the port that that robot used to reach the current node. $state$ indicates the current state of the robot: exploring, settled, or backtracking. $rnd\_cntr$ is used to keep track of the rounds.

\alglanguage{pseudocode}
\begin{algorithm}
	\caption{Rooted-Graph-LogN, run by each robot $u$}
	\label{prot:Rooted-Graph-LogN}
	\begin{algorithmic}[1]
	\State Initialize the following values: $port\_entered \leftarrow -1$, $parent\_ptr \leftarrow \bot$, $state \gets explore$.
	\For{$rnd\_cntr \gets 1, 2m$}
		\State Set $port\_entered$ to the port entered through.
		\If {$state = explore$}
			\If {Node has robot assigned to it}
				\State $state \leftarrow backtrack$. Move through $port\_entered$.
			\Else
				\If {$u$ is lowest label robot on node}
					\State $u$ assigns itself to node, $state \gets settled$, $parent\_ptr \leftarrow port\_entered$.
				\EndIf
				\If {$state \neq settled$}
					\State $port\_entered \leftarrow (port\_entered + 1) \mod$ degree of node.
					\If {$port\_entered = parent\_ptr$ of robot assigned to node}
						\State $state \leftarrow backtrack$.
					\EndIf
					\State Move through $port\_entered$
				\EndIf
			\EndIf
		\ElsIf{$state = backtrack$} 
			\State $port\_entered \leftarrow (port\_entered + 1) \mod$ degree of node.
			\If {$port\_entered \neq parent\_ptr$ of node}
				\State $state \leftarrow explore$.
			\EndIf
			\State Move through $port\_entered$.
		\EndIf
	\EndFor
	\Statex
\end{algorithmic}
\end{algorithm}

\begin{theorem} \label{the:rooted-graph-logn}
Algorithm \emph{Rooted-Graph-LogN} can be run by robots with $O(\log n)$ bits of memory to ensure dispersion occurs in $O(m)$ rounds on rooted graphs.
\end{theorem}

\begin{proof}
Since each robot performs the algorithm for $2m$ rounds, it is clear to see that running time is $O(m)$ rounds. Regarding memory, $port\_entered$, $parent\_ptr$, and $rnd\_cntr$ each take $O(\log n)$ bits of memory, $state$ takes two bits of memory (recall $O(\log m) = O(\log n^2) = O(\log n)$), and each robot needs $O(\log n)$ bits of memory to compare labels in order to decide which robot will be assigned to a given node. Thus every robot requires $O(\log n)$ bits of memory to run the algorithm.

We now need to show that dispersion is achieved in $2m$ rounds. Because every robot starts at the same node and follows the same path as other unsettled robots until it is assigned to a node, we just need to show that the last robot is settled within $2m$ rounds. Algorithm \emph{Rooted-Graph-LogN} represents a depth-first search of the graph and thus every edge is traversed at most twice. There are $m$ edges and thus within $2m$ rounds, all edges will be traversed at least once and by extension all nodes will be visited and all robots will be settled.
\end{proof}

\subsection{Dispersion on an Arbitrary Graph}\label{subsec:logn-disp-graph}

In this section, we describe an algorithm that can be run by robots with $O(\log n)$ bits of memory to achieve dispersion in $O(nm)$ rounds. 

The approach we use in this section bears similarity to the approach in \emph{Path-Ring-Tree-LogN} in Section~\ref{subsec:path-ring-tree-logn}. We still have every robot perform a DFS until it finds a node it can settle down at. However, the key to this approach is that we have the robot run multiple DFS's until it settles down. Consider any exploring robot $x$ with starting node label $a$. $x$ performs a DFS until it either settles down or comes across a settled robot $y$ whose starting node label $b$ is smaller than $a$. If $x$ comes across such a $y$, then $x$ sets its starting node label to $b$ and backtracks up the tree $y$ belongs to until it reaches $b$. Once it reaches $b$, it starts a new DFS and repeats the process above. 

However, if $x$ comes across a settled robot $z$ with starting node label $c$ where $c > a$, we run into the following problem. Namely, $z$ points towards the root of its tree and not the root of $x$'s tree, and thus interferes with $x$'s DFS. We fix this issue by having $z$ change its starting node label to $a$ and change its parent pointer to the port through which $x$ entered. Thus when two robots from different trees come into contact with each other, one of them converts to being a part of the tree of the other, either as a settled robot or as an exploring robot. Note that when multiple exploring robots converge on a node, either (i) one of them settles down if there is no settled robot and then converts the others to its tree or (ii) they all exchange information with each other and the settled robot to identify who has the tree with lowest label starting node and then the others convert to that tree.

An interesting note here is that an exploring robot that is backtracking (either just to its parent or to the root of the tree) may come across a settled robot with a smaller starting node label. In this case, the exploring robot simply changes trees yet again and backtracks to the root of the new tree. It can never be the case that a backtracking robot arrives at a node with a starting node label larger than its own. This is because whenever a settled robot changes its starting node label to something smaller, it will also re-orient its parent pointer. Hence, other robots backtracking along that node won't go to the old parent with a larger label but the new parent with a smaller one.

However, the above idea by itself does not guarantee that dispersion is achieved. In order to avoid cycles and successfully run the DFS's, we must have exploring robots somehow figure out when they've explored a back-edge of the tree they belong to; recall that cross-edges do not occur when using DFS on an undirected graph. Recall that a back-edge occurs when we move from a node $r$ distance from the root to a node $<r$ distance from the root. So if we have each settled robot maintain the label of its starting node and the distance from that node, exploring robots will know when they come across a back-edge and can then backtrack. Furthermore, since exploring robots also maintain the distance from their root, in case of tree conversions, the newly converted robot will also know its distance from its new root.

The main algorithm is \emph{Graph-LogN} which takes care of having the robot perform a DFS and check for cycles. It makes calls to the procedure \emph{Convert-If-Needed} at the beginning of each round in order to check if the robot needs to convert to a new tree. \emph{Convert-If-Needed} is run by a robot $u$ to see if it is co-located with robots having a different starting node. If so, it decides if $u$ should convert to the new tree, i.e. change its starting node to that of the new tree, or not. Additionally, if $u$ decides to convert and is an exploring or backtracking node, then it changes its state to $backtrack\_to\_root$ and backtracks to the root of the new tree. If $u$ decides to convert and is a settled node, then it updates its distance from the root to that of the robot that converted it plus $1$. The reason for the plus $1$ is that exploring robots update their own distance from root later in the algorithm so $u$ must do it manually here for itself.

\emph{Graph-LogN} uses six variables, $rnd\_cntr$, $port\_entered$, $parent\_ptr$, $state$, $starting\_node$, and $dist\_from\_start\_node$. $rnd\_cntr$ is used to denote the current round number. $port\_entered$ denotes the port through which a robot entered the current node. $parent\_ptr$ is used by a settled robot to indicate the port that that robot used to reach the current node. $state$ indicates the current state of the robot: exploring, settled, backtracking, or backtracking to root. $starting\_node$ denotes the robot assigned to the root node from which DFS starts. $dist\_from\_start\_node$ is used by to indicate the current distance of the robot from the starting node of the DFS.

\alglanguage{pseudocode}
\begin{algorithm}
	\caption{Convert-If-Needed, a procedure run by each robot $u$}
	\label{prot:Convert-If-Needed}
	\begin{algorithmic}[1]
		\State Communicate with other robots on same node. Find out $starting\_node$ values of all robots co-located with $u$ and let the least value heard be $least\_heard\_sn$. Let $x$ be the label of one of the robots with $starting\_node$ value $least\_heard\_sn$.
		\If {$u$'s $starting\_node$ value > $least\_heard\_sn$}
			\State $u$'s $starting\_node \gets least\_heard\_sn$.
			\If {$state = settled$}
				\State Change $parent\_ptr$ to port that $x$ entered node through.
				\State Set $distance\_from\_start\_node$ to that of $x$'s plus $1$. 
			\Else
				\State $state \gets backtrack\_to\_root$.
			\EndIf 
		\EndIf
		\Statex
	\end{algorithmic}
\end{algorithm}

\alglanguage{pseudocode}
\begin{algorithm}
	\caption{Graph-LogN, run by each robot $u$}
	\label{prot:Graph-LogN}
	\begin{algorithmic}[1]
		\State Initialize the following values: $port\_entered \gets -1$, $parent\_ptr \gets \bot$, $state \gets explore$, $starting\_node \gets u$, $dist\_from\_start\_node \gets -1$.
		\If{$u$ is lowest label robot on node}
			\State $state \gets settled$.
		\Else
			\State $starting\_node \gets$ label of lowest label robot.
		\EndIf
		
		\For{$rnd\_cntr \gets 1, 2mn + n^2$}
			\State \emph{Convert-If-Needed}.
			\If{$state = explore$}
				\State Set $port\_entered$ to the port entered through.
				\State Increment $dist\_from\_start\_node$.
				\If {Node is unassigned and $u$ is lowest label robot on node}
					\State $u$ assigns itself to node, $state \gets settled$, $parent\_ptr \gets port\_entered$.
				\EndIf
				\If {$state \neq settled$}
					\If {$u$'s $dist\_from\_start\_node >$ the settled robot's $dist\_from\_start\_node$}
						\State $state \gets backtrack$.
					\Else
						\State $port\_entered \gets (port\_entered + 1) \mod$ degree of node.
						\If {$port\_entered = parent\_ptr$ of robot assigned to node}
							\State $state \gets backtrack$.
						\EndIf
					\EndIf
					\State Move through $port\_entered$.
				\EndIf
			
			\ElsIf{$state = backtrack$}
				\State Decrement $dist\_from\_start\_node$.
				\State $port\_entered \gets (port\_entered + 1) \mod$ degree of node.
				\If {$port\_entered \neq parent\_ptr$ of node}
					\State $state \gets explore$.
				\EndIf
				\State Move through $port\_entered$.
				
			\ElsIf{$state = backtrack\_to\_root$}
				\If {the robot at current node is $starting\_node$}
					\State $state \gets explore$, $distance\_from\_start\_node \gets 0$, $port\_entered \gets 0$.
				\Else
					\State $port\_entered \gets parent\_ptr$ of robot assigned to node
				\EndIf
				\State Move through $port\_entered$.
			\EndIf
		\EndFor
		\Statex
	\end{algorithmic}
\end{algorithm}

\begin{theorem} \label{the:graph-logn}
	Algorithm \emph{Graph-LogN} can be run by robots with $O(\log n)$ bits of memory to ensure dispersion occurs in $O(mn)$ rounds on arbitrary graphs.
\end{theorem}

\begin{proof}
	Regarding memory complexity of robots, $rnd\_cntr$, $port\_entered$, $parent\_ptr$, $starting\_node$, and $dist\_from\_start\_node$ each take $O(\log n)$ bits of memory, $state$ takes two bits of memory, and each robot needs $O(\log n)$ bits of memory to compare labels in order to decide which robot will be assigned to a given node. Thus every robot requires $O(\log n)$ bits of memory to run the algorithm.
	
	We now need to show that dispersion is achieved in $2mn + n^2$ rounds. We show that any robot will eventually assign itself to a node within this many rounds. 
	
	Consider a robot $u$ performing a DFS on a graph. Due to the cycle checking strategy we implemented, no edge will be traversed more than twice and so within $2m$ rounds, all edges would be traversed if $u$ is allowed to complete its DFS. However, any time $u$ comes across a robot with a starting robot with a smaller label value, it must convert itself to this new tree, backtrack to the root of the tree and start performing a DFS all over again. It take no longer than $n$ rounds to backtrack to the root of any tree in the graph. Furthermore, there are at most $n-1$ other labels that $u$ may see and thus have to convert to. Thus after at most $n*(2m + n)$ rounds, $u$ would be able to complete a DFS without encountering any new trees to convert to. Thus, after that many rounds, $u$ would have successfully found a node to settle at. This holds for all such $u$ and thus after $2mn + n^2$ rounds, dispersion is achieved.
\end{proof}

\subsection{Dispersion on a Rooted Tree}\label{subsec:rooted-tree-logn}
A rooted tree is a tree where all robots start at one node called the root. We look at dispersion on a rooted tree and show that when the memory of each robot is $O(\log n)$ bits, we can achieve dispersion in $O(D^2)$ rounds. 

We define a \textit{fully dispersed subtree rooted at a node} to be a subtree with that node as the root where every node has a robot assigned to it. We use the term \textit{settled} robot to denote a robot which is assigned to a node.

Our algorithm works in stages. Each stage corresponds to (i) the unassigned robots located at the root exploring the tree one level deeper than the previous stage and assigning themselves to empty nodes when possible, and (ii) settled robots updating their information about which ports lead to fully dispersed subtrees. Once a settled robot decides that the subtree rooted at the node it's been assigned to has been fully dispersed, it stops executing the algorithm. 

Each robot is in one of four states: (i) \textit{root}, (ii) \textit{explore}, (iii) \textit{wait}, or (iv) \textit{settled}. Only one robot will be in state $root$ and this is the robot that has been assigned to the root of the tree. This robot executes the algorithm until the subtree rooted at its node is fully dispersed.

A robot in the $settled$ state is essentially assigned to a given node. It maintains a pointer to the port it used to reach this node initially. It also maintains a counter to count how many of the other ports of the node have been completely dispersed. We say that a \textit{port has been completely dispersed} when the subtree rooted at the node attached to the port has been fully dispersed. For the remainder of this section, we use level and depth to mean the same thing.

For a given node, define an \textit{unassigned leaf} of the node as a descendant $v$ of the node in the tree such that (i) $v$ is unassigned and (ii) $v$'s immediate parent is assigned a robot. If the given node is itself unassigned a robot but its parent is assigned a robot, we define the number of unassigned leaves of the node as $1$. If the given node and its parent are both unassigned robots, we define its number of unassigned leaves as $0$.

The algorithm uses nine variables, $port\_entered$, $parent\_ptr$, $state$, $dist\_from\_root$, $fully\_dispersed\_port$, $fully\_dispersed$, $num\_robots\_reqd$, $lvl\_cntr$, and $rnd\_cntr$. $port\_entered$ denotes the port through which a robot entered the current node. $parent\_ptr$ is used by a settled robot to indicate the port that that robot used to reach the current node. $state$ indicates the current state of the robot: exploring, settled, or backtracking. $dist\_from\_root$ indicates the distance of the current node from the root. $fully\_dispersed\_port$ is a counter used by settled robots to count how many of their ports leads to fully dispersed subtrees or not. $fully\_dispersed$ is a boolean variable used by settled robots indicating whether the subtree rooted at the robot's node is fully dispersed or not. $num\_robots\_reqd$ is used by a settled robot to indicate how many robots are currently needed to fill its unassigned leaves. $lvl\_cntr$ is used to denote the current maximum level to which robots will explore the tree. $rnd\_cntr$ is used to denote the current round for a given stage.

Every stage $lvl\_cntr$ consists of $2*lvl\_cntr + 1$ rounds. The first $lvl\_cntr$ rounds are used by robots in state $explore$ to traverse the tree to depth $lvl\_cntr-1$ with the help of procedure \emph{Further-Explore}. \emph{Further-Explore} guarantees the following property: in the given stage, if a given node has $x$ unassigned leaves, exactly $x$ robots will be sent to the node in the course of the exploration. This is achieved as follows through \emph{Further-Explore}. For a given node $v$ with children that are settled, in the round where robots reach $v$, the robots assigned to $v$'s children would have moved up to $v$ as well and transmitted the number of unassigned leaves (i.e. number of robots required to be sent through each port). In the subsequent round, these children of $v$ return to their assigned nodes and the robots in state $explore$ know exactly which port to traverse. Once robots reach the parent of an unassigned node, one robot goes to each unassigned node and gets assigned to that node using procedure \emph{Robot-Assignment}.

A robot in $explore$ state not sent from the root changes its state to $wait$ and waits for the beginning of the next stage. Just before the start of a new stage, all nodes in state $wait$ change their state to $explore$. 

The second $lvl\_cntr$ rounds of a stage are used by robots in state $settled$ to propagate information up to their parents about whether the subtree rooted at the corresponding node is fully dispersed or not. This propagation of information occurs in a bottom up fashion with robots at level $lvl\_cntr - 1$ moving up to inform robots at their parents about this information. Subsequently robots at level $lvl\_cntr -1$ move up and inform their parents and so on.

The last round in the stage is used by robots at depth $1$ to move to the $root$ in preparation for the next stage.

\alglanguage{pseudocode}
\begin{algorithm}
	\caption{Robot-Assignment, a procedure run by each robot $u$}
	\label{prot:Robot-Assignment}
	\begin{algorithmic}[1]
		\State Let $\delta$ be degree of node.
		\If {$state = root$}
			\State $num\_robots\_reqd \gets \delta$.
		\Else
			\State $num\_robots\_reqd \gets \delta - 1$, $parent\_ptr \gets port\_entered$.
		\EndIf
		\State Set entries in $fully\_dispersed\_port$ from index $\delta$ to $\Delta - 1$ to $1$.
		\Statex
	\end{algorithmic}
\end{algorithm}

\alglanguage{pseudocode}
\begin{algorithm}
	\caption{Further-Explore, a procedure run by each robot $u$}
	\label{prot:Further-Explore}
	\begin{algorithmic}[1]
		\State Let $\delta$ be degree of node. Order ports in increasing order of port numbers excluding node's robot's $parent\_ptr$ port and any fully dispersed ports. Every robot on the node knows this same order $p_1, \ldots, p_k$ now.
		\State Coordinate with other robots to calculate $u$'s position in the order of increasing labels of exploring robots.
		\If {Nodes from one level down are present on this node and communicating their $num\_robots\_reqd$ values 
			\phantom . \phantom . }
			\State  Robots are assigned to ports in increasing label order until all port's $num\_robots\_reqd$ requirements \phantom . \phantom . \phantom . are filled. Remaining robots stay unassigned.
		\Else
			\State Robots are assigned to ports in increasing label order such that one robot is assigned per port. \newline \phantom . \phantom . \phantom .  Remaining robots stay unassigned.
		\EndIf

		\If {$u$ was assigned to a given port $p_i$}
			\State $u$ moves through port $p_i$.
			\State Increment $dist\_from\_root$.
		\Else
			\State $state \gets wait$.		
		\EndIf
		
		\Statex
	\end{algorithmic}
\end{algorithm}

\alglanguage{pseudocode}
\begin{algorithm}
	\caption{Rooted-Tree-LogN, run by each robot $u$}
	\label{prot:Rooted-Tree-LogN}
	\begin{algorithmic}[1]
		\State Initialize the following values: $port\_entered \gets 0$, $parent\_ptr \gets \bot$, $state \gets explore$, $dist\_from\_root \gets 0$, $fully\_dispersed\_port \gets 0$, $fully\_dispersed \gets FALSE$, $num\_robots\_reqd \gets 0$, $lvl\_cntr \gets 0$.
		
		\If {$u$ is lowest label robot on node}
			\State $state \gets root$.
			\State \emph{Robot-Assignment}.
		\EndIf
		
		\While {$fully\_dispersed = FALSE$} 
			\State Increment $lvl\_cntr$.
			\For {$rnd\_cntr \gets 1, 2 * lvl\_cntr$ + 1}
				\If {$state = root$}
					\If {A robot enters through a given port and its $fully\_dispersed$ bit $= TRUE$}
						\State Increment $fully\_dispersed\_port$.
						\State If $fully\_dispersed\_port = $ degree of node, set $fully\_dispersed \gets TRUE$.
					\Else
						\State Do nothing this round.
					\EndIf
		
				\ElsIf {$state = explore$}
					\State Set $port\_entered \gets$ port entered through.
					\If {No robot assigned to node}
						\State $state \gets settled$
						\State \emph{Robot-Assignment}.
					\EndIf
					\If {$state \neq settled$}
						\State \emph{Further-Explore}.
					\EndIf
		
				\ElsIf {$state = wait$}
					\State Do nothing for round.
					\If {$rnd\_cntr = 2*lvl\_cntr + 1$}
						\State $state \gets explore$
					\EndIf
		
				\ElsIf {$state = settled$}
					\If {($rnd\_cntr = dist\_from\_root - 1$) OR ($rnd\_cntr = 2*lvl\_cntr - dist\_from\_root$) OR \newline \phantom . \phantom . \phantom .
						\phantom . \phantom . \phantom .
						\phantom . \phantom . \phantom . \phantom .
						 ($dist\_from\_root = 1$ AND $rnd\_cntr = 2*lvl\_cntr + 1$ AND $fully\_dispersed = FALSE$) \newline \phantom . \phantom . \phantom . \phantom . \phantom . \phantom . \phantom . \phantom . \phantom . \phantom .}
						\State If $fully\_dispersed\_port = $ degree of node $-1$, set $fully\_dispersed \gets TRUE$.
						\State Move through $parent\_ptr$ port.
					\ElsIf {($rnd\_cntr = dist\_from\_root$) OR ($rnd\_cntr = 2*lvl\_cntr - dist\_from\_root + 1$)}
						\State Set $port\_entered \gets$ port entered through.
						\State Communicate with robot(s) at current node and inform them about $u$'s $fully\_dispersed$,
						\newline \phantom . \phantom .\phantom .  \phantom . \phantom .\phantom . \phantom . \phantom .\phantom . \phantom . \phantom . \phantom .
						 $port\_entered$, and $num\_robots\_reqd$ values.
						\State Move through $port\_entered$ port.
					\Else
						\State If any robots inform $u$ about $fully\_dispersed$, increment $fully\_dispersed\_port$. 
						\State If robots inform $u$  about their $num\_robots\_reqd$ values, update $u$'s $num\_robots\_reqd$ value \newline \phantom .  \phantom . \phantom . \phantom . \phantom . \phantom . \phantom . \phantom . \phantom . \phantom .  with the sum of all values heard.
					\EndIf
				\EndIf
			\EndFor
		\EndWhile
		\Statex
	\end{algorithmic}
\end{algorithm}

\begin{theorem} \label{the:rooted-tree-logn}
Algorithm \emph{Rooted-Tree-LogN} can be run by robots with $O(\log n)$ bits of memory to ensure dispersion occurs in $O(D^2)$ rounds on rooted trees.
\end{theorem}

\begin{proof}
We first show that algorithm \emph{Rooted-Tree-LogN} only requires every robot to have $O(\log n)$ bits of memory and then we subsequently prove the correctness and running time guarantees of the algorithm. Every robot requires $O(\log \Delta)$ bits of memory for $port\_entered$, $O(\log \Delta)$ bits of memory for $parent\_ptr$, $2$ bits of memory for $state$, $O(\log D)$ bits of memory for $dist\_from\_root$, $O(\log n)$ bits of memory for $fully\_dispersed\_port$, $1$ bit of memory for $fully\_dispersed$, $O(\log D)$ bits for $rnd\_cntr$, $O(\log n)$ bits for $num\_robots\_reqd$, $O(\log D)$ bits of memory for $lvl\_cntr$, and $O(\log n)$ bits of memory to compare labels when deciding which robots will be assigned to a given node. Thus each robot requires $O(\log n)$ bits of memory in order to successfully run \emph{Rooted-Tree-LogN}.

We prove the correctness of the algorithm by showing that the following induction hypothesis holds true for all $1 \leq i \leq d+1$.

\begin{lemma}\label{lem:stage-clear}
At the end of stage $i$ of \emph{Rooted-Tree-LogN}, all nodes at depth~$\leq i-1$ have robots assigned to them, for all $1 \leq i \leq d+1$.
\end{lemma}

\begin{proof}
At the end of stage $1$, the only node at depth $0$, the root, has a robot assigned to it. So the base case of the induction holds true.

Let us assume that the induction hypothesis holds true for some stage $i$, i.e. at the end of stage $i$ all nodes at depth~$\leq i-1$ have robots assigned to them. We now must show that at the end of stage $i+1$, all nodes at depth $\leq i$ have robots assigned to them. We show this by showing that for every unassigned leaf of the root at the beginning of stage $i+1$, we send exactly one robot to it by the end of stage $i+1$. This implies that all nodes at depth $i$ at the end of stage $i+1$ will have robots assigned to them. Note that once a node has a robot assigned to it, that robot will never be unassigned from that node, so all nodes at previous depth levels continue to have robots assigned to them.

In a given stage $i$, the first $i$ rounds are used for exploration and the second $i$ rounds are used to update robots assigned to nodes about the number of their assigned leaves. Thus at the end of stage $i$, each robot knows the total number of assigned leaves it has, even if it doesn't know which ports lead to how many unassigned leaves. During the first $i+1$ rounds of stage $i+1$, in every round $r$ the following three types of robots will be present in nodes at depth $r-1$: (i) exploring robots, (ii) robots assigned to nodes at depth $r-1$, and (iii) any robots assigned to children of nodes at depth $r-1$ if those children themselves have unassigned leaves. 

For any node at depth $r-1$ at round $r$, one of two cases may occur. Either its children are all unassigned nodes or they are assigned nodes. In the former case, one robot moves down each port leading to an unassigned node. In the latter case, the robots are able to communicate with each other in round $r$ and coordinate such that the number of robots that move down each port is equivalent to the exact number of unassigned leaves of the node attached to that port.

Thus for a given stage $i+1$, after $i+1$ rounds, all nodes at levels $\leq i$ have robots assigned to them.
\end{proof}

Thus, for a depth $d$ tree, after $d+1$ stages, all nodes will have robots assigned to them. Furthermore, for a given node $v$, if the subtree rooted at that node is fully dispersed by the end of a given stage $i$, then by the end of stage $i$ the robot assigned to $v$ will be present on $v$ and will stop executing the algorithm. Thus at the end of stage $d+1$, all robots stop executing the algorithm and dispersion is achieved.

As for running time, since there are $d+1$ stages and each stage takes $O(d)$ rounds, each robot executes a total of $O(d^2) = O(D^2)$ rounds. Thus the running time of \emph{Rooted-Tree-LogN} is $O(D^2)$ rounds.
\end{proof}



\section{Dispersion with $O(n \log n)$ Bits of Memory}
\label{sec:n-logn-memory}
 
We look at dispersion on an arbitrary graph and present an $O(m)$ round algorithm to achieve it when the memory of each robot is $O(n \log n)$ bits. We use ideas from the algorithms \emph{Rooted-Graph-LogN} and \emph{Graph-LogN} described in Sections~\ref{subsec:logn-disp-rooted-graph} and~\ref{subsec:logn-disp-graph} respectively. 

\subsection{Dispersion on a Graph}
\alglanguage{pseudocode}
\begin{algorithm}
	\caption{Graph-N-LogN, run by each robot $u$}
	\label{prot:Graph-N-LogN}
	\begin{algorithmic}[1]
		\State Initialize the following values: $port\_entered \leftarrow -1$, $state \gets explore$, $labels\_seen\_so\_far\_with\_parent\_ptr[n]$.
		\If{$u$ is lowest label robot on node}
			\State $state \gets settled$.
		\Else
			\State Add $<u,\bot>$ to $labels\_seen\_so\_far\_with\_parent\_ptr$.
		\EndIf
		
		\For {$rnd\_cntr \gets 1, 2m$}
			\State Set $port\_entered$ to the port entered through.
			\If {$state = explore$}
				\If {Node has robot $v$ assigned to it  \newline \phantom . \phantom . \phantom . \phantom . \phantom . \phantom . and $v$ in  $labels\_seen\_so\_far\_with\_parent\_ptr$   }
					\If {Node is unassigned and $u$ is lowest label robot on node}
						\State $u$ assigns itself to node, $state \gets settled$.
					\Else
						\State $state \leftarrow backtrack$. Move through $port\_entered$.
					\EndIf
				\Else
					\If{Node has robot $v$ assigned to it}
						\State Add $<v, port\_entered>$ to~$labels\_seen\_so\_far\_with\_parent\_ptr$.
					\Else
						\If {$u$ is lowest label robot on node}
							\State $u$ assigns itself to node, $state \gets settled$.
						\Else 
							\State Add $<$ lowest label seen among robots on node, \newline \phantom . \phantom . \phantom . \phantom . \phantom . \phantom . \phantom . \phantom . \phantom . \phantom . \phantom . \phantom . \phantom . $port\_entered>$   to $labels\_seen\_so\_far\_with\_parent\_ptr$.
						\EndIf
					\EndIf
					\If {$state \neq settled$}
						\State $port\_entered \leftarrow (port\_entered + 1) \mod$ degree of node.
						\If {$port\_entered = $ parent pointer of robot assigned to node \newline \phantom . \phantom . \phantom . \phantom . \phantom . \phantom . \phantom . \phantom .  \phantom . \phantom .  according to   $labels\_seen\_so\_far\_with\_parent\_ptr$}
							\State $state \leftarrow backtrack$.
						\EndIf
						\State Move through $port\_entered$.
					\EndIf
				\EndIf
			\ElsIf{$state = backtrack$}
				\State $port\_entered \leftarrow (port\_entered + 1) \mod$ degree of node.
				\If {$port\_entered \neq$ parent pointer of robot assigned to node according \newline \phantom . \phantom . \phantom . \phantom . \phantom . to  $labels\_seen\_so\_far\_with\_parent\_ptr$}
					\State $state \leftarrow explore$.
				\EndIf
				\State Move through $port\_entered$.
			\EndIf
		\EndFor
		\Statex
	\end{algorithmic}
\end{algorithm}

We have robots perform a DFS while checking for cycles. If we ensure that we provide the robots with enough time to eventually visit every node, then we are guaranteed that each robot will eventually find an empty node to settle down at. We can use the cycle checking idea from \emph{Rooted-Graph-LogN}. In that algorithm, robots checked if they were exploring a cycle by checking if the current node had been previously visited. Due to limited memory, our cycle checking mechanism in that algorithm relied on identifying a previously visited node by the fact that a robot was already assigned to it. By allowing each robot to have $O(n \log n)$ bits of memory, robots can check for up to $n$ nodes that they've already visited before. Since there are totally $n$ nodes, we are guaranteed to never land up in a cycle.

The algorithm uses four variables, $port\_entered$, $state$, $labels\_seen\_so\_far\_with\_parent\_ptr[n]$, and $rnd\_cntr$. $port\_entered$ denotes the port through which a robot entered the current node. $state$ indicates the current state of the robot: exploring, settled, or backtracking. $labels\_seen\_so\_far\_with\_parent\_ptr[n]$ is an array of size $n$ which is used by exploring robots to store the labels of robots assigned to nodes they've already traversed in the DFS as well the port they used to enter that node for the first time. $rnd\_cntr$ is used to keep track of the rounds. 

\begin{theorem} \label{the:graph-n-logn}
Algorithm \emph{Graph-N-LogN} can be run by robots with $O(n\log n)$ bits of memory to ensure dispersion occurs in $O(m)$ rounds on rooted graphs.
\end{theorem}

\begin{proof}
Every robot is active for $2m$ rounds, so it is clear to see that the running time is $O(m)$ rounds. $port\_entered$ and $rnd\_cntr$ each take $O(\log n)$ bits of memory (recall $O(\log m) = O(\log n^2) = O(\log n)$), $state$ takes one bit of memory, $labels\_seen\_so\_far\_with\_parent\_ptr[n]$ takes $O(n \log n)$ bits of memory, and each robot needs $O(\log n)$ bits of memory to compare labels in order to decide which robot will be assigned to a given node. Thus every robot requires $O(n \log n)$ bits of memory to run the algorithm.

We now need to show that dispersion is achieved in $2m$ rounds. Recall that each robot is performing its own independent DFS. Our goal is to now bound the time taken to perform this DFS. 

Consider a given robot $u$. The labels of robots stored in its $labels\_seen\_so\_far\_with\_parent\_ptr$ correspond to nodes that $u$ has passed through every port of. Thus \emph{Graph-N-LogN} represents a depth-first search of the graph and thus every edge is traversed at most twice. There are $m$ edges and thus within $2m$ rounds, all edges will be traversed at least once and by extension all nodes will be visited and all robots will be settled. 
\end{proof}


\section{Conclusions and Future Work}

\label{sec:conc}

We proposed a new problem, called dispersion, for mobile robots on graphs which is closely related to the problems of scattering on graphs, collective mobile robot exploration, and load balancing on graphs. We provided a lower bound for the memory required by each robot to achieve dispersion. We also developed algorithms to solve this problem for various types of graphs given different constraints on memory. We list open problems of interest below.\\
\textbf{Open Problem 1:} Can we develop a $o(n)$ round algorithm to achieve dispersion in arbitrary trees?\\
\textbf{Open Problem 2:} Can we develop an algorithm to achieve dispersion in $o(m)$ rounds on an arbitrary graph or else strengthen our lower bound?


\section*{Acknowledgements}
We would like to thank Vishvajeet Nagargoje for useful discussions at various stages of this work. We would like to thank Narayanaswamy N. S. for his useful insight which lead to the improvement in the memory complexity for the rooted tree algorithm.

\bibliographystyle{abbrv}
\bibliography{references} 

\end{document}